\documentclass[11pt]{article}

\usepackage{array,fullpage,multirow}
\usepackage{amsmath,amsthm,amsfonts}
\usepackage{amssymb,graphicx}
\usepackage{graphicx}
\usepackage{url}
\usepackage{latexsym}
\usepackage[colorlinks=true,linkcolor=blue,citecolor=red]{hyperref}

\usepackage{verbatim}

\newtheorem{theorem}{Theorem}[section]

\newtheorem*{remark}{Remark}
\newtheorem{proposition}[theorem]{Proposition}
\newtheorem{lemma}[theorem]{Lemma}

\newtheorem{claim}[theorem]{Claim}

\newtheorem{corollary}[theorem]{Corollary}
\newtheorem{definition}[theorem]{Definition}
\newtheorem{problem}[theorem]{Problem}

\newcommand{\cP}{\mathcal{P}}
\newcommand{\F}{\mathcal{F}}
\newcommand{\maxCSP}[1][3]{\textrm{max-CSP}}

\newcommand{\gapCSPF}{\textrm{Gap-CSP-}\F}
\newcommand{\maxCSPF}[1][3]{\textrm{Max-CSP-}\F}
\newcommand{\CSP}[1][3]{\textrm{CSP}}
\newcommand{\kCSP}[1][3]{#1\textrm{-CSP}}
\newcommand{\kSAT}[1][3]{#1\textrm{-SAT}}
\newcommand{\kCNF}[1][3]{#1\textrm{-CNF}}
\newcommand{\GapColoring}{\textrm{Gap-Coloring}}
\renewcommand{\SS}[1][3]{\textrm{Subset-Sum}}
\newcommand{\GapVC}{\textrm{Gap-Vertex-Cover}}
\newcommand{\HamCycle}{\textrm{HamCycle}}

\newcommand{\NP}{\ensuremath{\mathcal{NP}}}

\newcommand{\RP}{\ensuremath{\mathcal{RP}}}
\newcommand{\coRP}{\ensuremath{\rm co}\mathcal{RP}}
\newcommand{\BPP}{\ensuremath{\mathcal{BPP}}}

\newcommand{\N}{{\mathbb N}}

\newcommand{\E}{{\mathop{\mathbf{E}}}}

\newcommand{\eps}{\varepsilon}
\newcommand{\seq}{\subseteq}

\newcommand{\val}{val}

\begin{document}

\title{On Percolation and $\NP$-Hardness}

\author{Daniel Reichman\thanks{
Department of Computer Science, Cornell University, Ithaca, NY, USA.  Email: \texttt{daniel.reichman@gmail.com}. Supported in part by NSF grants IIS-0911036 and CCF-1214844, AFOSR grant FA9550-08-1-0266, and ARO grant W911NF-14-1-0017}
\and
Igor Shinkar\thanks{
    Courant Institute of Mathematical Sciences,
    New York University.
    Research supported by NSF grants CCF 1422159, 1061938, 0832795  and Simons Collaboration on Algorithms and Geometry grant.
    } }

\maketitle
\begin{abstract}
We consider the robustness of computational hardness of problems
whose input is obtained by applying independent random deletions to worst-case instances.
For some classical $\NP$-hard problems on graphs, such as Coloring, Vertex-Cover, and Hamiltonicity, we examine the complexity of these problems when edges (or vertices) of an arbitrary
graph are deleted independently with probability $1-p > 0$. We prove that for $n$-vertex graphs, these problems remain
as hard as in the worst-case, as long as $p > \frac{1}{n^{1-\epsilon}}$ for arbitrary $\epsilon \in (0,1)$, unless $\NP \subseteq \BPP$.

We also prove hardness results for Constraint Satisfaction Problems,
where random deletions are applied to clauses or variables,
as well as the $\SS$ problem, where items of a given instance are deleted at random.
\end{abstract}

\newpage
%%%%%%%%%%%%%%%%%%%%%%%%%%%%%%%%%%%%%%%%%%%%%%%%%%%%%%%
\section{Introduction}
%%%%%%%%%%%%%%%%%%%%%%%%%%%%%%%%%%%%%%%%%%%%%%%%%%%%%%%

The theory of $\NP$-hardness suggests that we are unlikely to find optimal
or near optimal solutions to $\NP$-hard problems in polynomial time.
This theory applies to worst-case settings where one considers the worst
running-time over the worst possible input. It is less clear whether these
hardness results apply to ``real-life" instances. One way to address this
question is to examine to what extent known $\NP$-hardness results are
stable under random perturbations, as it seems reasonable to assume that
a given instance of a problem may be subjected to noise originating from
multiple sources.

In this work we study \emph{worst-case} instances that are subjected to
\emph{random perturbations} of a specific type, namely, random deletions.
We focus on the following deletion process, known as \emph{edge percolation}:
given a graph $G$ consider a random subgraph of $G$ obtained
by deleting each edge of $G$ \emph{independently} with probability $1-p$
(where $p \in (0,1)$ may depend on the size of the instance).
This model generalizes familiar random graph models such as the
Erd\"{o}s-R\'{e}nyi random graph $G(n,p)$. Instead of focusing on deleting
edges at random from the complete graph, our starting graph $G$ may be chosen
arbitrarily out of all $n$-vertex graphs. Then the edges of $G$ are deleted
independently with probability $1-p$. The case of \emph{vertex} percolation, where the
vertices of a given graph are deleted independently at random is examined as well.
We also study random deletions in other $\NP$-complete problems, such as $\kSAT[3]$ and $\SS$.

Throughout we refer to instances that are subjected to
random deletions as \emph{percolated instances}.
Our main question is whether such percolated instances
remain hard to solve by polynomial-time algorithms, under reasonable assumptions from complexity theory.

\paragraph{A first example.}
Consider the 3-Coloring Problem, where given a graph $G$
we need to decide whether $G$ is $3$-colorable.
Suppose we sample a random subgraph $G'$ of $G$, by deleting each
edge of $G$ independently with probability $1/2$,
and ask whether the resulting graph is 3-colorable (one can prove similar results when edges are deleted with probability smaller than $1/2$, but we focus on the
case where $p=1/2$ for concreteness).
Is there a polynomial time algorithm can decide with high probability whether $G'$ is $3$-colorable?
Or does the problem remain hard in the sense that an efficient algorithm that determines whether $G'$ is $3$-colorable
would imply that every problem in $\NP$ admits an efficient algorithm?

We demonstrate that a polynomial-time algorithm that decides whether $G'$
is $3$-colorable is unlikely. We show it by considering the following polynomial
time reduction from the 3-Coloring Problem to itself.
Given an $n$-vertex graph $H$ the reduction outputs a graph $G$ that is
an $R$-blow-up of $H$ for $R=C\sqrt{\log(n)}$ where $C>0$ is large enough.
That is, each vertex of $H$ is replaced by a cloud of $R$ vertices that form
an independent set in $G$, and each edge of $H$ is replaces with a complete
$R \times R$ bipartite graph in $G$ between the corresponding clouds in $G$.
It is clear that $H$ is 3-colorable if and only if $G$ is 3-colorable.

In fact, the foregoing reduction satisfies a stronger robustness property
for the random subgraph $G'$ of $G$. Namely, if $H$ is 3-colorable,
then $G$ is 3-colorable, and hence $G'$ is also 3-colorable with probability 1.
On the other hand, if $H$ is not 3-colorable, then $G$ is not 3-colorable, and with high
probability $G'$ is not 3-colorable either.
Indeed, for any edge $(v_1,v_2)$ in $H$ let $U_1,U_2$ be two clouds in $G$
corresponding to $v_1$ and $v_2$. Fixing two arbitrary sets $U_1' \seq U_1$
and $U_2' \seq U_2$ each of size at least $R/3$, the probability there is no
edge connecting a  vertex from $U_1$ to a vertex in $U_2$ is at most
$2^{-R^2/9} = 2^{-C' \log n}$. By union bound we get that
with high probability (over the sampling of $G'$) for any two clouds $U_1,U_2$
corresponding to an edge in $H$ and any
$U_1' \seq U_1$ and $U_2' \seq U_2$ each of size at least $R/3$
there is at least one edge between $U'_1$ and $U'_2$.
Therefore, with high probability any 3-coloring
of $G'$ can be decoded to a 3-coloring of $H$ by coloring each vertex
$v$ of $H$ with the color that appears the largest number of times in the coloring of
the corresponding cloud in $G'$ (breaking ties arbitrarily).
This suggests that unless $\NP \seq \coRP$ there is no
polynomial time algorithm that given a 3-colorable graph $G$ finds a legal
3-coloring of a random subgraph of $G$ obtained by subsampling every edge with probability $1/2$.

%%%%%%%%%%%%%%%%%%%%%%%%%%%%%%%%%%%%%%%%%%%%%%%%%%%%%%%
\subsection{Motivation}
%%%%%%%%%%%%%%%%%%%%%%%%%%%%%%%%%%%%%%%%%%%%%%%%%%%%%%%

There is a large body of research dealing with computational problems on
random graphs and formulas \cite{FriezeM97}. This study has resulted with
several algorithms which have proven effective on random instances.
A more recent line of research suggests that efficient algorithms for finding exact or approximate
solutions to computational problems on random objects may not exist~\cite{achi,Coja-OghlanE11,Rossman14}. Other works have demonstrated that assuming problems
on randomly generated instances to be hard, implies hardness of approximation
results for certain optimization problems that are not known to follow from
worst-case assumptions~\cite{Feige02}. These results raise the question of
what kind of hardness results for solving optimization problems exactly or
approximately for percolated instances can be derived when the original
instance is selected in a worst-case fashion.
We note that proving hardness results for our model
should be an easier task than proving hardness results for random instances
such as those arising, for example, from the Erd\"{o}s-R\'{e}nyi random graph
$G(n,p)$, as we have more freedom in choosing the instance that is subjected
to random deletions.

The study of random discrete structures has resulted with a wide range of mathematical tools which have proven instrumental in proving rigorous results regarding such structures \cite{Random_Graphs,FriezeM97,grimmett,mezard}. Our hybrid model may offer the opportunity to apply these methods to a broader range of distributions of instances of $\NP$-hard problems.

%%%%%%%%%%%%%%%%%%%%%%%%%%%%%%%%%%%%%%%%%%%%%%%%%%%%%%%
\subsection{Our results}
%%%%%%%%%%%%%%%%%%%%%%%%%%%%%%%%%%%%%%%%%%%%%%%%%%%%%%%

We consider several classical $\NP$-hard problems, for which we prove that
they remain hard also on percolated instance. Unless stated otherwise, $n$ stands for the number of vertices in the graph.

\begin{itemize}
\item
For the Maximum Independent Set problem we use the hardness of approximation result
of~\cite{FeigeKilian} to show that for edge percolation, where we keep each edge of a
given graph with probability $p > \frac{1}{n^{1-\eps}}$ for some $\eps  \in (0,1)$
it is hard to approximate the maximal independent set on percolated instances
within any factor better than $\Omega(\frac{1}{pn^{1-\epsilon}})$
We also show that the chromatic number of a percolated instance in
hard to approximate within $O(pn^{1-\eps})$.
Note that for $p > \frac{1}{n^{1-\eps}}$
(in fact, for $p > \frac{C \log(n)}{n}$)
such random percolated graphs have maximal degree at most $O(pn)$
with high probability,
and hence can be colored efficiently using $O(pn)$ colors.

We also prove that for vertex deletion these problems remain as hard to approximate as in
the worst-case, as long as the vertices remain in the graph independently
with probability $p > \frac{1}{n^{1-\eps}}$ for some $\eps \in (0,1)$.
More specifically, denoting by $m$ the number of remaining vertices
in the vertex percolated subgraph, it is hard to approximate its
chromatic number or independence number
within a factor of $m^{1-\delta}$ (resp. $\frac{1}{m^{1-\delta}}$) for arbitrary constant $\delta \in (0,1)$.

\item For the Vertex-Cover problem, we prove that for any constant $\delta > 0$ an algorithm
that gives $2-\delta$ approximation for percolated instances implies also
a $2-2\delta$ approximation algorithm for worst-case instances.
Our results hold for both edge and vertex percolation, where the
edges or the vertices of a given graph remain with probability
$p>\frac{1}{n^{1-\eps}}$ for some $\eps \in (0,1)$.
In particular, assuming the Unique Games Conjecture, the results of ~\cite{Khot} imply 
there is no randomized polynomial time algorithm that with high probability
gives $2-\delta$ approximation for the Vertex Cover problem on percolated instances.

\item For the Hamiltonicity problem, we prove hardness results for percolated
instances with respect to edge percolation on directed graphs.
We show that the problem where one needs to determine whether a graph contains
a Hamiltonian cycle is also hard for percolated graphs,
where each edge of a given graph is kept in the graph with probability
$p > \frac{1}{n^{1-\eps}})$ for any $\eps \in (0,1)$.

\item We also consider percolation of $\kSAT[3]$ instances where clauses are deleted at random with probability $p$.
We prove that, unless $\NP \seq \coRP$ for every $\eps,\delta \in (0,1)$ if
clauses of a given $\kSAT[3]$ survive with probability $p > \frac{1}{n^{2-\delta}}$,
then $(7/8+\eps)$-approximation on percolated instances
is hard, as it is the case for worst-case instances.
This result is nearly tight, as known algorithms for random $\kSAT[3]$ formulas imply that
for sufficiently small $c > 0$ if clauses survive deletions with probability
$p > \frac{c}{n^2}$, the resulting formula admits a satisfying assignment
which can be found efficiently with high probability (see the related works section for more details).

More generally, we prove that
unless $\NP \subseteq \BPP$ arbitrary $k$-ary Boolean $\CSP$ problems
are as hard to approximate on percolated instances as in the worst-case,
as long as each clause is percolated with probability $p > \frac{1}{n^{k-1-\eps}}$
for any $\eps \in (0,1)$.

The key step in the proof is establishing that any hardness of approximation of
a $k$-ary $\CSP$ problems can be translated to the same hardness approximation
on instances whose number of constraints
is $n^{k-\eta}$ for arbitrary small $\eta>0$. For example, relying on the result of
H{\aa}stad~\cite{Hastad01} we show that $\kSAT[3]$ is $\NP$-hard to approximate with
a ratio better than $7/8+\epsilon$ even on instances that
contain at least $n^{3-\eta}$ clauses.

We also consider variable percolation, where each variable is deleted with probability
$p>\frac{1}{n^{1-\eps}}$ for any $\eps \in (0,1)$
(when a variable is removed all clauses containing it are removed as well).
Similar ideas as those applied for the clause percolation case
imply that such percolated instance are essentially as hard as in the worst case.

\item We study percolation on instances of the $\SS$ problem,
where each item of the set is deleted with probability $1-p$.
We show that the problem remains hard as long as $p=\Omega (\frac{1}{n^{1/2-\eps}})$ for some $\eps \in (0,1/2)$,
where $n$ is the number of items in the given instance.

\end{itemize}

%%%%%%%%%%%%%%%%%%%%%%%%%%%%%%%%%%%%%%%%%%%%%%%%%%%%%%%
\subsection{Our techniques}
%%%%%%%%%%%%%%%%%%%%%%%%%%%%%%%%%%%%%%%%%%%%%%%%%%%%%%%
In proving hardness results for percolated instances we use the concept of
\emph{robust reductions} which we explain next.
It will be convenient to consider promise problems.
Recall, that a promise problem is a generalization of a decision problem,
where for the problem $L$ there are two disjoint subsets $L_{YES}$ and $L_{NO}$,
such that an algorithm that solves $L$ must accept all the inputs in $L_{YES}$
and reject all inputs in $L_{NO}$.
If the input does not belong to $L_{YES} \cup L_{NO}$, there is no requirement on the output
of the algorithm.

\begin{definition}
	For each $y \in \{0,1\}^*$ let $Perc(y)$ be a distribution on $\{0,1\}^*$,
	that is samplable in time that is polynomial in $|y|$.

	For two promise problems $A =(A_{YES},A_{NO})$ and $B = (B_{YES},B_{NO})$
    a polynomial time reduction $r$ from $A$ to $B$ is said to be
    \emph{$Perc$-robust} if
	\begin{enumerate}
		\item For all $x \in A_{YES}$ it holds that $r(x) \in B_{YES}$, and $\Pr[ Perc(r(x)) \in B_{YES} ] > 1 - o(1)$.
		\item For all $x \in A_{NO}$ it holds that $r(x) \in B_{NO}$, and $\Pr[ Perc(r(x)) \in B_{NO} ] > 1 - o(1)$.
	\end{enumerate}

	If in the first item we have $\Pr[Perc(r(x)) \in B_{YES}] = 1$,
	then we say that $r$ is a $Perc$-robust $\coRP$-reduction.
	Similarly, if in the second item we have $\Pr[Perc(r(x)) \in B_{NO}] = 1$,
    then we say that $r$ is a $Perc$-robust $\RP$-reduction.
\end{definition}

Let us elaborate on how robust reductions apply when graph percolation is concerned.
Let $A,B$ be $\NP$ languages over graphs, and given a graph $y$,
let $Perc(y)$ be vertex or edge percolation of $y$.
In such setting a reduction $r$ is said to be $Perc$-robust
if it satisfies the standard definition of a reduction,
i.e., $x \in A$ if and only if $r(x) \in B$, and \emph{in addition}
the containment of $r(x)$ in $B_{YES}$ (or in the $B_{NO}$) is robust
to random deletions that are captured by the distribution $Perc(r(x))$.
As a concrete example, consider the language $\kSAT[3]$ consisting of all satisfiable $\kCNF[3]$ formulas,
and the language $\HamCycle$ consisting of graphs containing a Hamiltonian cycle.
Consider a reduction from $\kSAT[3]$ to $\HamCycle$ that given a $\kCNF[3]$ formula $\phi$ produces a graph $G$.
Let $G_p$ be a random subgraph of $G$ obtained from $G$ by including each edge of $G$ independently
with probability $p$. The reduction is said to be robust with respect to edge percolation
if the following two assertions hold:
(1) if $\phi$ is satisfiable, then $G$ contains a Hamiltonian cycle and
$G_p$ contains a Hamiltonian cycle with high probability,
and (2) if $\phi$ is not satisfiable, then $G$ is not Hamiltonian,
and with high probability $G_p$ is not Hamiltonian either.

We make two remarks regarding the example above.
First note that if $G=(V,E)$ does not contain a Hamiltonian cycle,
then neither does any graph $G'=(V,E')$ where $E' \subseteq E$. Therefore, if such robust
reduction exists, then it is necessarily a robust $\RP$-reduction.

Note also that such reduction must be such that
if $\phi$ is satisfiable, then $G$ contains a Hamiltonian cycle,
and furthermore $G$ must contain many Hamiltonian cycles,
even if $\phi$ has only a small (e.g., constant number of satisfying assignments.
Indeed, if $G$ contained only $K$ Hamiltonian cycles for some constant $K$,
then, $G_p$ is unlikely to be Hamiltonian, as typically only $pn$ edges ($n$ is the number of vertices of $G$)
of each cycle will remain after percolating the edges.
That is, such a reduction \emph{cannot} be a parsimonious reduction
in the sense that the reduction preserves the number of $\NP$-witnesses.

\medskip

The existence of such a reduction implies that the Hamiltonicity problem is
in some sense $\NP$-hard on percolated instances.
Below we explain this hardness more precisely.
We start with the following definition.

\begin{definition}
	Let $L =(L_{YES},L_{NO})$ be a promise problem,
	and for each $y$ instance of $L$, let $Perc(y)$ be a distribution
	on instances of $L$	that is samplable in time that is polynomial in $|y|$.

	The problem $L =(L_{YES},L_{NO})$ is said to be \emph{$\NP$-hard
    under a $Perc$-robust reduction} if there exists a \emph{$Perc$-robust}
    reduction from an $\NP$-hard problem to $L$.
\end{definition}

We use the term $Perc$-robust to avoid confusion with other notions of robust reductions that have appeared in the literature. In order to ease readability, we will often write
robust reductions instead, always refereing to $perc$-robust reductions as defined above.

\begin{proposition}\label{prop:BPP RP harndess}
	Let $L =(L_{YES},L_{NO})$ be a promise problem,
	and for each $y$ instance of $L$, let $Perc(y)$ be a distribution
	on instances of $L$	that is samplable in time that is polynomial in $|y|$.

	If $L$ is $\NP$-hard under a $Perc$-robust reduction,
	then there is no polynomial time algorithm that when given an input $y$
	decides with high probability whether $Perc(y) \in B$, unless $\NP \seq \BPP$.

	If the foregoing hardness holds under a $Perc$-robust $\RP$-reduction ($\coRP$-reduction),
    then the same conclusion holds, unless $\NP \seq \RP$ (resp. $\NP \seq \coRP$).
\end{proposition}

For $\coRP$-reduction we have the following search decision version of the foregoing proposition.
\begin{proposition}\label{prop:coRP search}
	Let $L =(L_{YES},L_{NO})$ be a promise problem,
	and for each $y$ instance of $L$ let $Perc(y)$ be a distribution
	on instances of $L$	that is samplable in time that is polynomial in $|y|$.

    If $L$ is $\NP$-hard under a $Perc$-robust reduction,
    then, there is no polynomial time algorithm that when given an input $y$
	with high probability finds a witness for the assertion $Perc(y) \in B$, unless $\NP \seq \coRP$.
\end{proposition}

An example of an application of Proposition~\ref{prop:coRP search}, consider the $\kSAT[3]$ problem.
Recall that by a result of H{\aa}stad~\cite{Hastad01}
given a satisfiable $\kSAT$ instance $\Phi$ it is $\NP$-hard to find an assignment that satisfies
significantly more than $7/8$ fraction of the constraints of $\Phi$.
A stronger conclusion follows from Theorem~\ref{thm:3SAT clause perc}.
Namely, given a satisfiable $\kSAT$ instance $\Phi$ it is hard to find an assignment
that satisfies significantly more than $7/8$ fraction of the constraints in a random subformula of $\Phi$,
obtained from $\Phi$ be deleting each clause with probability, say, $p=1/2$
(while, any assignment that satisfies $\Phi$, also satisfies every subformula of $\Phi$).

To construct robust reductions we use two methods.
One is to apply hardness of approximation results implied by the PCP Theorem~\cite{AS98,ALMSS98}.
Intuitively, the gap between YES-case and NO-case in such hardness results,
makes it possible to prove that percolated instances remain hard
as random deletions will not affect the optimum by much, keeping
(with high probability) the distinction between the YES-case and the NO-case.

When known hardness results do not suffice (as it is the case for the Vertex Cover problem),
or when hardness of approximation results are unlikely,
(as is the case for the $\SS$ problem which admits a PTAS) we ``blowup'' the instance
in a certain way and prove that this blowup preserves certain combinatorial properties
even when edges (or vertices) are deleted with high probability.
The most standard blowup technique is to replace, given a graph $G$, every vertex of $G$
with a large independent set and connect two independent sets that correspond to adjacent
vertices of $G$ by a complete bipartite graph.
This method has been previously used to prove the $\NP$-hardness of Feedback Arc Set on tournaments \cite{Ailon}.
Other variants of blowup are used for the Hamiltonian cycle problem and $\SS$.

%%%%%%%%%%%%%%%%%%%%%%%%%%%%%%%%%%%%%%%%%%%%%%%%%%%%%%%
\subsection{Related Works}\label{sec:related}
%%%%%%%%%%%%%%%%%%%%%%%%%%%%%%%%%%%%%%%%%%%%%%%%%%%%%%%

Randomly subsampling subgraphs by including each edge independently in the sample with probability $p$ has been studied extensively in works concerned with cuts and flows (e.g., \cite{karger}). More recently, sampling subgraphs has been used to find independent sets \cite{FeigeR15} (the main sampling technique used, e.g., the layers model is not independent-it introduces dependencies between sampled vertices). The effect of subsampling variables in mathematical relaxations of constraint satisfaction problems on the value of these relaxations was studied in \cite{barak}. Edge-percolated graphs have been also used to construct hard-instance for specific algorithms. For example, ~\cite{Kucera} proved that the well known greedy coloring algorithm performs poorly on the complete $r$-partite graph in which every edge is removed independently with probability $1/2$ and $r=n^{\epsilon}$ for $\epsilon>0$. Namely, for this graph $G$, even if vertices are considered in a random order by the greedy algorithm, with high probability $\Omega(\frac{n}{\log n})$ colors are used to color the percolated graph whereas $\chi(G) \leq n^{\epsilon}$.

The work of~\cite{gamarnik} examined the problem of finding a maximum independent set in regular graphs where the weights of the vertices are i.i.d. exponential random variables.
In this work the authors prove that for $3$-regular graphs, it is the case that for every $\epsilon$ the problem admits a $(1-\epsilon)$ approximation in time $n^kg(\epsilon)$ where $k$ is independent of $n$ and $g(\epsilon)$ depends only on $\epsilon$. They also prove that for large enough (constant) degree $\Delta$, the problem of approximating the expected size of a maximum independent set in such randomly weighted graphs is essentially as hard as solving MIS on graphs with maximal degree $\Delta$. Our hardness results are based on different ideas than those of \cite{gamarnik}.

The field of stochastic optimization is concerned with solving computational problems where elements of the instance (e.g., weights, the existence of edges) are random variables. Typically, the main focus in this line of works is to design algorithms that make decisions (at least for part of the input) before the random variables have been instantiated, with good \emph{expected} guarantees (e.g., \cite{Dean,KleinbergRT00}). Our work is exclusively concerned with fully instantiated problems. In addition, we focus on a very specific type of uncertainty, where every random variable is either zero or one. As a result, the sampled objects admit a straightforward combinatorial interpretation (e.g., randomly sampled subgraphs or formulas) that is lacking when considering random variables such as the exponential distribution. In addition, dealing with a restricted family of random variables makes it more challenging to prove hardness results regarding instances with edges or vertices whose weights are distributed as these random variables.

\medskip

When $p$ is sufficiently small, algorithms for random graphs and random formulas can be proven to find the optimal solution (with high probability) for percolated instance. For example, for graph coloring, it is known that for $p=\frac{1+\eps}{n}$ with some positive constant $\eps>0$,
with high probability $G(n,p)$ is $2$-degenerate, and hence can be 3-colored in polynomial time~\cite{L}. Since the property of being 2-degenerate is monotone, and as $2$-colorability can be decided in polynomial time, it follows that for every $n$-vertex graph and $p \leq \frac{1+\epsilon}{n}$, one can find in polynomial time with high probability a coloring of the edge percolated graph with the minimum number of colors.

Similar reasoning applies to $\kSAT[3]$ formulas. It is well known that there exists $c>0$ such that a random $\kSAT[3]$ formula in which each possible clause is added independently with probability $p=\frac{c}{n^2}$ can be solved with high probability using the pure literal heuristic~\cite{broder}. As observed in~\cite{broder}, if this heuristic fails in finding a satisfying assignment for a formula $\phi$ it will still fail to find a satisfying assignment if clauses are added to $\phi$. This implies that for any $n$-variable $\kSAT[3]$ formula $\Phi$, if $p \leq \frac{c}{n^2}$, then the clause-percolated formula is satisfiable and furthermore a satisfying assignment can be found in polynomial time using the pure literal heuristic.

%%%%%%%%%%%%%%%%%%%%%%%%%%%%%%%%%%%%%%%%%%%%%%%%%%%%%%%
\subsection{Preliminaries}
%%%%%%%%%%%%%%%%%%%%%%%%%%%%%%%%%%%%%%%%%%%%%%%%%%%%%%%

In this work, we will only consider simple graphs without multiple edges and self loops. When directed graphs are concerned we allow the two directed edges $(u,v)$ and $(v,u)$ to coexist-such a situation is not considered as having multiple edges.

Given a graph $G = (V,E)$ (that may be directed or undirected) and $p \in (0,1)$,
we denote by $G_{p,e} = (V,E')$ the probability space of graphs on the same set of vertices,
where each edge $e \in E$ is contained in $E'$ independently with probability $p$.
We will say that $G_{p,e}$ is obtained by edge percolation.
We define $G_{p,v} = (V',E')$ as the probability space of graphs,
in which every vertex $v \in V$ is contained in $V'$ independently with probability $p$,
and $E'$ is the subgraph of $G$ induced by the vertices $V'$.
We will sometime say that $G_{p,v}$ is obtained from $G$ by vertex percolation.
When dealing the running time on percolated instances we will always measure running time
in terms of the size of the percolated instance. For edge percolation, it makes little difference
as far as polynomial-time algorithms are concerned, as the percolated and original graphs have the same number of vertices.
For vertex percolation, this is not the case, since for tiny values of $p$
the size of the percolated graph will be typically much smaller than the size of original graph.
In this work we will be only dealing with the case where $p=\frac{1}{n^{1-\Omega(1)}}$,
hence with high probability the size of the percolated and the original graphs
are polynomially related as well.

Given a graph property $\cP$ and a sequence of probability distributions $(\mu_n)_{n}$ over $n$-vertex graphs,
we will say that $\cP$ holds with high probability if $\lim_{n \to \infty} \Pr_{G \sim \mu_n} [ G \in \cP ]=1.$

We say that an algorithm approximates a maximization problem within a ratio of $0 < a \leq 1$
(where $a$ to depend on the size of the instance)
if it returns a feasible solution that is at least $a \cdot OPT$, where $OPT$ is the value of
the optimal solution.
Similarly, we say that an algorithm approximates a minimization problem within a ratio of $b \geq  1$,
if it returns a feasible solution that is at most $b \cdot OPT$ where $OPT$ is the value of
the optimal solution.

We shall rely on the following version of the Chernoff bound
(see, e.g.,~\cite{Vazirani}).
\begin{theorem}[Multiplicative Chernoff bound]
Let $X_1,\dots,X_n$ be independent 0-1 random variables with $\Pr[X_i=1] = p$.
Then,
\[
	\Pr[|\sum_{i=1}^n X_i - pn | \geq \eps pn] \leq e^{-C \eps^2 pn},
\]
for some absolute constant $C > 0$.
\end{theorem}

\begin{corollary}~\label{cor:concentration for m sequences}
Let $X_1^{(1)},\dots,X_n^{(1)}, \dots, X_1^{(m)},\dots,X_n^{(m)}$ be independent 0-1 random variables with $\Pr[X_i^{(j)}=1] = p$.
Then, for some absolute constant $C > 0$ it holds that
\[
	\Pr[\exists j \in [m]:|\sum_{i=1}^n X_i^{(j)} - pn | \geq \sqrt{C pn \log(m)}] \leq m^{-3}.
\]
\end{corollary}

\begin{proof}
	By the multiplicative Chernoff bound above for each $j \in [m]$ it holds that
	$\Pr[|\sum_{i=1}^n X_i^{(j)} - pn | \geq \sqrt{C pn \log(m)}]  \leq e^{-C \log(m)} < m^{-4}$,
	where $C > 0$ is some absolute constant.
	Therefore,
\begin{eqnarray*}
	\Pr[\exists j \in [m]:|\sum_{i=1}^n X_i^{(j)} - pn | \geq \sqrt{C pn \log(m)}]
	& = & 1 - \Pr[\forall j \in [m]:|\sum_{i=1}^n X_i^{(j)} - pn | \leq \sqrt{C pn \log(m)}]  \\
	& \leq & 1 - (1 - m^{-4})^m \\
	& \leq & m^{-3},
\end{eqnarray*}
as required.
\end{proof}
%%%%%%%%%%%%%%%%%%%%%%%%%%%%%%%%%%%%%%%%%%%%%%%%%%%%%%%
\section{Graph Coloring, Independent Set and Percolation}
%%%%%%%%%%%%%%%%%%%%%%%%%%%%%%%%%%%%%%%%%%%%%%%%%%%%%%%

An independent set in a graph $G=(V,E)$ is a set of vertices that spans no edge.
The independence number of a graph is the size of an independent set of maximum size.
Given a graph $G$, we denote the independence number of $G$ by $\alpha(G)$.
A legal coloring of a graph $G$ is an assignment of colors to vertices
that no two adjacent vertices have the same color. The chromatic number
of $G$, denoted by $\chi(G)$ is the minimum number of colors that allows
a legal coloring of $G$. Clearly, $\chi(G) \cdot \alpha(G) \geq n$.

In this section we demonstrate the hardness of approximating $\alpha(G)$ and $\chi(G)$
in percolated graphs for both edge and vertex percolation. We start with edge percolation.

%%%%%%%%%%%%%%%%%%%%%%%%%%%%%%%%%%%%%%%%%%%%%%%%%%%%%%%
\paragraph{Edge percolation}
%%%%%%%%%%%%%%%%%%%%%%%%%%%%%%%%%%%%%%%%%%%%%%%%%%%%%%%
Here a crucial observation is that a graph
without large independent sets cannot contain large sets that span a too small number of edges.
We need first the following lemma, due to Turan (see, e.g.~\cite{AlonSpencer}).

\begin{lemma}\label{lem:is rand}
Every graph $H$ with $l$ vertices and $e$ edges contains an independent set of size at least $\frac{l^2}{2e+l}$.
\end{lemma}

\begin{corollary}\label{cor:small is many edges}
Let $G=(V,E)$ be an $n$-vertex graph satisfying $\alpha(G)<k$.
Then every set of vertices of size $l \geq k$ spans at least $l(l-k)/2k$ edges.
\end{corollary}

\begin{proof}
Let $H$ be a subgraph of $G$ induced by $l$ vertices, and suppose that $H$ spans $e$ edges.
Then, by Lemma~\ref{lem:is rand} we have $\alpha(H) \geq \frac{l^2}{2e+l}$.
On the other hand, $\alpha(H) \leq \alpha (G) \leq k$, and hence $\frac{l^2}{2e+l} \leq k$, as required.
\end{proof}

\begin{lemma}\label{lem:independent_perc}
Let $G=(V,E)$ be an $n$-vertex graph.
Then, with high probability $\alpha(G_{p,e})\leq O\left(\frac{\alpha(G)}{p}\log (np)\right)$.
\end{lemma}
\begin{proof}
Let $k=\alpha(G)+1$. Let $C>0$ be a large enough constant.
By the corollary above, every set of size $l=C\frac{\alpha(G)}{p}\log (np)$,
spans at least $\frac{l(l-k)}{2k}$ edges. Hence,
by taking union bound over all subsets of size $l$, the probability there exists
a set of size $l$ in $G_p$ that spans no edge is at most
\[
	{n \choose l}  \cdot (1-p)^{\frac{l(l-k)}{2k}}
	< \left( \frac{e n}{l} \right) ^l \cdot \exp \left(- p \cdot \frac{l(l-k)}{2k} \right)
	< (np)^{- \Omega(l)},
\]
where the last inequality uses the choices of $l$ and $k$, implying that $\left( \frac{e n}{l} \right)^l < (np)^l$
and $\exp( - p \frac{l(l-k)}{2k}) < \exp(- \Omega(l \cdot \log(np))) = (np)^{-\Omega(l)}$.
\end{proof}

We observe that in general, the upper bound above cannot be improved,
as it is well known that the independence number of $G(n,p)$ is
$O\left(\frac{\log (np)}{p}\right)$ with high probability (see, e.g.,~\cite{Random_Graphs}).

\medskip

We are now ready to prove that it is hard to approximate the independence number
and the chromatic number on edge percolated graphs.
For this we consider the following gap problem which we call $\GapColoring(\chi,\alpha)$,
where the YES-instances are all graphs $G$ with $\chi(G) \leq \chi$
and the NO-instances are all graphs $G$ with $\alpha(G) \leq \alpha$.
(We assume that for $n$-vertex graphs $G$
the parameters of the problem  are such that $\alpha(G) \cdot \chi(G) < n$,
so that the YES-instances and the NO-instances are disjoint sets.)

\begin{theorem}\label{thm:MIS edge perc}
Let $\eps \in (0,1)$ be a fixed constant.
Let $p>\frac{1}{n^{1-2\eps}}$, and
let $\chi = n^\eps$ and $\alpha = \frac{n^\eps}{p}$,
where $n$ denotes the number of vertices in a graph.
Then, the $\GapColoring(\chi, \alpha)$ problem is $\NP$-hard under a robust reduction
with respect to edge percolation with parameter $p$.

In particular, unless $\NP \subseteq \BPP$ there is no polynomial time algorithm that approximates either
$\alpha(G_{p,e})$ or $\chi(G_{p,e})$ within a factor $\frac{1}{pn^{1-2\epsilon}}$ (resp. $pn^{1-2\epsilon}$).
\end{theorem}

\begin{proof}
By a result of Feige and Kilian \cite{FeigeKilian}, for any $\eps > 0$ it is $\NP$-hard to decide whether a
given $n$-vertex graph $G = (V,E)$ satisfies $\chi(G)\leq n^{\eps}$ or $\alpha(G) \leq n^{\eps/2}$.
Let $\tilde{G} = G_{p,e}$ be the $p$-edge percolated subgraph of $G$.
Next we claim the following.

\begin{description}
   \item[YES-case:] If $\chi(G) \leq n^{\eps}$, then $\chi(\tilde{G}) < n^\eps$.
   \item[NO-case:] If $\alpha(G) \leq n^{\eps/2}$, then $\alpha(\tilde{G})  < \frac{n^{\eps}}{p}$.
\end{description}

The YES-case is clear, since $\tilde{G}$ is obtained from $G$ by removing edges which can only decrease the chromatic number.

For the NO-case suppose that $\alpha(G) \leq n^{\eps/2}$. Then, by Corollary~\ref{lem:independent_perc}
it follows that with high probability
$\alpha(G_{p,e})\leq O\left(\frac{\alpha(G)}{p}\log (np)\right) < \frac{n^{\eps}}{p}$,
as required.

The ``in particular'' part of the theorem
follows from the fact that an $n$-vertex graph $G$ it holds that $\chi(G) \cdot \alpha(G) \geq n$.
\end{proof}

\begin{remark}
Note that for constant $p>0$ (e.g., $p=1/2$) this theorem establishes
an inapproximability for the independence number of $G_{p,e}$,
that matches the inapproximability for the worst case.
\end{remark}

\begin{remark}
Note also that for $p > \frac{1}{n^{1-\eps}}$
(in fact, for $p > \frac{C \log(n)}{n}$) such random percolated
graphs have maximal degree at most $O(pn)$ with high probability.
Therefore, such graphs $\tilde{G}$ can be colored efficiently using $O(pn)$ colors.
In particular, with high probability $\tilde{G}$ contains an independent set of size
$\Omega(1/p)$ and hence, $\alpha(\tilde{G})$ can be approximated within a factor
of $1/pn$ on $p$-percolated instances.
\end{remark}

%%%%%%%%%%%%%%%%%%%%%%%%%%%%%%%%%%%%%%%%%%%%%%%%%%%%%%%
\paragraph{Vertex percolation}
%%%%%%%%%%%%%%%%%%%%%%%%%%%%%%%%%%%%%%%%%%%%%%%%%%%%%%%
We now move on to deal with vertex percolation.
We show that approximating the $\alpha(G)$ and $\chi(G)$ on percolated instances
is essentially as hard as worst-case instances,
even if vertices remain with probability $\frac{1}{n^{1-\eps}}$,
where $n$ is the number of vertices in the graph for any $\eps \in (0,1)$.
We do it again by proving hardness of the gap problem $\GapColoring$ for percolated instances.

Note that in the case of vertex percolation, we (in)approximablity guarantee
should depend on the number of vertices in the percolated graph $G_{p,v}$,
and not in the original graph.

\begin{theorem}\label{thm:MIS vertex perc}
Let $\eps,\delta \in (0,1)$ be fixed constants.
Then, for any $p>\frac{1}{n^{1-\delta}}$
the $\GapColoring(\chi, \alpha)$ problem is $\NP$-hard under a robust reduction
with respect to vertex percolation with parameter $p$,
where $\chi = m^{\eps/2}$ and $\alpha = m^{\eps/2}$,
with $m$ denoting the number of vertices in the vertex percolated graph.

In particular, unless $\NP\subseteq\BPP$ there is no polynomial time algorithm that approximates either
$\alpha(G_{p,v})$ or $\chi(G_{p,v})$ within a factor $m^{1-\eps}$ for constant any $\eps>0$.
\end{theorem}

\begin{proof}
For a given $p>\frac{1}{n^{1-\delta}}$ let $c = \frac{\log(pn)}{\log(n)} \in (\delta, 1)$
be  such that $p = \frac{1}{n^{1-c}}$, and let $\eta = \eps \cdot c/3$.
By a result of Feige and Kilian~\cite{FeigeKilian}, it is $\NP$-hard to decide whether a
given $n$-vertex graph $G = (V,E)$ satisfies $\chi(G)\leq n^{\eta}$ or $\alpha(G) \leq n^{\eta}$.

Let $\tilde{G} = G_{p,v}$ be the $p$-vertex percolated subgraph of $G$,
and let $m$ be the number of vertices in $\tilde{G}$.
By concentration bounds, we have $|m - pn| < 0.1 pn$ with high probability,
and we shall assume from now on that this is indeed the case.
By the choice of the parameters this implies $n^{\eta}  < m^{\eps/2}$
%\\DETAILS: $n^{\eta} = (n^{2c/3})^{\eps/2} < (0.9 n^c)^{\eps/2} = (0.9 \cdot np)^{\eps/2} < m^{\eps/2}$ \\
Therefore, if $\chi(G) \leq n^{\eta}$, then $\chi(\tilde{G}) \leq n^{\eta} < m^{\eps}$.
On the other hand, if $\alpha(G) \leq n^{\eta}$, then $\alpha(\tilde{G}) < n^{\eta} < m^{\eps}$,
and the proof follows.
\end{proof}

%%%%%%%%%%%%%%%%%%%%%%%%%%%%%%%%%%%%%%%%%%%%%%%%%%%%%%%
\subsection{Vertex Cover}
%%%%%%%%%%%%%%%%%%%%%%%%%%%%%%%%%%%%%%%%%%%%%%%%%%%%%%%

An vertex cover in a graph $G = (V, E)$ is a set of vertices $S \seq V$ such that
every edge $e \in E$ is incident to at least one vertex in $S$.
that spans no edge. In the Minimum Vertex Cover problem we are given a graph $G$
and our goal is to find a vertex cover of $G$ of minimal size.

Note that for an $n$-vertex graph $G$ it holds that $G$ contains a vertex cover of size $k$
if and only if it contains an independent set of size $n-k$.

There is a simple factor 2 approximation algorithm
for the Minimum Vertex Cover problem~\cite{Vazirani}.
On the hardness side, the problem is
$\NP$-hard to approximate within a factor of 1.3606~\cite{DinurSafra05},
and assuming the Unique Games Conjecture is known to be
$\NP$-hard to approximate within a factor of $(2-\eps)$
for any constant $\eps>0$ ~\cite{Khot}.
We prove that the same hardness result hold also
when instead of worst-case instances one considers

We will need the following definition.
\begin{definition}
	Given a graph $G$, the $R$-blowup of $G$ is a graph $G'=(V',E')$,
	where every vertex $v$ is replaced by an independent set $\widetilde{v}$ of size $R$,
	which we also call the cloud corresponding to $v$.
	If $(u,v) \in E$, then $\widetilde{u}$ and $\widetilde{v}$
	are connected by a complete $R \times R$ bipartite graph.
\end{definition}

%%%%%%%%%%%%%%%%%%%%%%%%%%%%%%%%%%%%%%%%%%%%%%%%%%%%%%%
\paragraph{Edge percolation}
%%%%%%%%%%%%%%%%%%%%%%%%%%%%%%%%%%%%%%%%%%%%%%%%%%%%%%%

We have the following simple lemma regarding independent
sets in edge percolated subgraph of $K_{R,R}$.

\begin{lemma}\label{lemma:edge_perc K_RR}
Consider the complete bipartite graph $K_{R,R}$ with bipartition $A,B$,
and let $G_{p,e}$ be the edge percolation of $K_{R,R}$ with probability $p$.
Then, the probability that there is an independent set $I$ in $G_{p,e}$
such that $|I \cap A| = |I \cap B| = C \log(R)/p$ is at most $R^{-3}$,
where $C$ is a large enough constant independent of $n$ or $p$.
\end{lemma}

\begin{proof}
For fixed sets $S_A \seq A$ and $S_B \seq B$ each of size $C \log(R)/p$
the probability that $S_A$ and $S_B$ span no edge is $(1-p)^{(C \log(R)/p)^2}$.
Therefore, by union bound over all $S_A$ and $S_B$ the probability
that there is is an independent set $I$ in $G_{p,e}$
with $|I \cap A| = |I \cap B| = C \log(R)/p$ is at most
\[
	{R \choose C \log(R)/p}^2 (1-p)^{(C \log(R)/p)^2} \leq m^{2C \log(R)/p}e^{-p(C \log(R)/p)^2}
\]
which is at most $R^{-3}$ for large enough $C$.
\end{proof}

Consider the following $\GapVC(c,s)$ problem
where the YES-instances are graphs that have a vertex cover of size $c n$,
and NO-instances are all graphs whose minimal vertex cover is larger than $s n$,
where $n$ is the number of vertices in $G$.
Note that, equivalently,
the YES-instances are graphs that contain an independent set of size $\alpha(G) \geq (1-c)n$,
the NO-instances are graphs whose maximal independent set is of size $\alpha(G) \leq (1-s)n$.

We remark that the result of Khot and Regev~\cite{Khot} proves that assuming the Unique Games
Conjecture the problem $\GapVC(1-\eps,1/2+\eps)$ if $\NP$-hard for all constant $\eps >0$.
We use this fact in order to hardness of approximation for this problem on percolated instances.

\begin{theorem}\label{thm:VC edge perc}
Let $\eps,\delta \in (0,1)$ be fixed constants.
Assuming the Unique Games Conjecture, $\GapVC(1-\eps,1/2+\eps)$
is $\NP$-hard under a robust reduction with respect to edge percolation with parameter $p$
for any $p > \frac{1}{n^{1-\delta}}$, where $n$ denotes the number of vertices in the given graph.

In particular, assuming the Unique Games Conjecture $(2-\eps)$-approximation of the
Vertex Cover problem is hard on edge percolated instances.
\end{theorem}

\begin{proof}
By~\cite{Khot} assuming the Unique Games Conjecture, for any $\eps > 0$ the
problem $\GapVC(1-\eps,1/2+\eps)$ is $\NP$-hard. Equivalently, given an $N$-vertex
graph $G$ is is $\NP$-hard to distinguish between the case that
$\alpha(G) > (1/2-\eps)N$ and the case that $\alpha(G) < \eps N$.
We show a reduction from this problem to itself
(with slightly larger parameter $\eps$) that is robust for edge percolation.

Consider the reduction that given a graph $G$ outputs
the $R$-blowup of $G$, which we denote by $H$, with $R > n$ to be chosen later.
That is the graph $H$ is a graph on $n = NR$ vertices, and it is clear that $\alpha(H) = \alpha(G) \cdot R$.
Therefore, this is indeed a reduction from the $\GapVC(1-\eps,1/2+\eps)$ to itself.
We show below that in fact the reduction is robust for edge percolation.
In order to do it we prove that with high probability
\begin{equation}\label{eq:VC edge perc}
	\alpha(G) \cdot R \leq \alpha(\widetilde{H}) \leq \alpha(G) \cdot R+(C\log(R)/p) \cdot N,
\end{equation}
where $\widetilde{H} = H_{p,e}$ denotes the edge percolation of $H$ with parameter $p$.
Indeed, the left inequality is clear because $\alpha(\widetilde{H}) \geq \alpha(H) = \alpha(G) \cdot R$,
since $\widetilde{H}$ is a subgraph of $H$.

For the right inequality, by Lemma~\ref{lemma:edge_perc K_RR}
with probability at least $(1 - N^2/R^3)$ the following holds:
for every edge $(u,v)$ of $G$ the corresponding clouds $\widetilde{u}$ and $\widetilde{v}$ in $H$
are such that there is no independent set $I$ in $\widetilde{H}$,
such that $|I \cap \widetilde{u}| \geq C \log(R)/p$
and $|I \cap \widetilde{v}| \geq C \log(R)/p$.
Therefore, if $I$ is an independent set that intersects
some clouds on more than $C\log(R)/p$, then the vertices corresponding
the these clouds must form an independent set in $G$.
Thus, with probability at least $(1 - N^2/R^3)$ we have
$\alpha(\widetilde{H})\leq \alpha(G) \cdot R+(C\log(R)/p) \cdot N$.

Next we choose the parameter $R$ such that the reduction above
is indeed a robust reduction for edge percolation with parameter $p$.
For the parameter $p$ let $c = \frac{\log(pn)}{\log(n)}$,
and let $R = N^{2/c}$ (where $N$ is the number of vertices in the original graph).

Now, if $\alpha(G) > (1/2-\eps) N$, then by~\eqref{eq:VC edge perc} we have
$\alpha(\widetilde{H}) \geq \alpha(G) \cdot R > (1/2-\eps) NR = (1/2-\eps) n$,
and hence $\widetilde{H}$ contains a vertex cover of size $(1/2+\eps) n$
On the other hand, we claim that if $\alpha(G) < \eps N$, then with high probability
$\alpha(\widetilde{H})  < 2\eps n$.
Indeed, by the choice of $R$ we have
$p = \frac{1}{n^{1-c}}	=  \frac{1}{(NR)^{1-c}}	>  \frac{N}{\alpha(G)} \cdot \frac{C \log(R)}{R}$.
Therefore, by the right inequality of~\eqref{eq:VC edge perc} we have
$\alpha(\widetilde{H}) \leq \alpha(G) \cdot R + (C\log(R)/p) \cdot N \leq 2 \alpha(G) \cdot R < 2\eps n$,
and hence $\widetilde{H}$ does not have a vertex cover of size $(1-\eps) n$.
This completes the proof of Theorem~\ref{thm:VC edge perc}.
\end{proof}

%%%%%%%%%%%%%%%%%%%%%%%%%%%%%%%%%%%%%%%%%%%%%%%%%%%%%%%
\paragraph{Vertex percolation}
%%%%%%%%%%%%%%%%%%%%%%%%%%%%%%%%%%%%%%%%%%%%%%%%%%%%%%%
We proceed with vertex percolation.
Note that when considering vertex percolation,
the percolation parameter $p$ depends on the
number of vertices in the given (worst-case instance) graph,
while the performance of the algorithm is measured with
respect to the number of vertices in the percolated graph,
which is close to $pn$ with high probability

\begin{theorem}\label{thm:VC vertex perc}
Let $\eps,\delta \in (0,1)$ be fixed constants.
Assuming the Unique Games Conjecture, $\GapVC(1-\eps,1/2+\eps)$
is $\NP$-hard under a robust reduction with respect to vertex percolation with parameter $p$,
for any $p > \frac{1}{n^{1-\delta}}$, where $n$ is the number of vertices in the given graph.

In particular, assuming the Unique Games Conjecture $(2-\eps)$-approximation of the
Vertex Cover problem is hard on vertex percolated instances.
\end{theorem}

\begin{proof}
The reduction is the same as in the proof of Theorem~\ref{thm:VC edge perc}.
For the parameters $p$ and $\eps$ let $c = \frac{\log(pn)}{\log(n)}$
so that $p = \frac{1}{n^{1-c}}$, and let $R = (\frac{N}{\eps^2})^{1/c}$.
Given a graph $G$ the reduction produces the $R$-blowup of $G$,
which we denote by $H$. That is the graph $H$ is a graph on $n = NR$ vertices.

Let $\widetilde{H} = H_{p,e}$ denote the vertex percolation of $H$ with parameter $p$.
By Corollary~\ref{cor:concentration for m sequences}, with high probability
the number of vertices in $\widetilde{H} = H_{p,e}$, which we denote by
$m$ is between $pNR - C\sqrt{pNR \log(NR)}$ and $pNR + C\sqrt{pNR \log(NR)}$,
and the number of vertices in every cloud of $\widetilde{H}$ is between
$pR-C\sqrt{pR\log N}$ and $pR+C\sqrt{pR\log N}$,
for some absolute constant $C>0$ independent of $N$ or $p$.

%Since the graph $\widetilde{H}$ is obtained from still a blowup of $G$ (possibly with clouds of different sizes),
Clearly any independent set $I$ in $\widetilde{H}$ gives rise to an independent set in $G$
by taking all vertices $v$ of $G$ such that $I$ intersects the corresponding cloud $\widetilde{v}$.
This implies that with high probability it holds (for $N$ large enough) that
\[
	\alpha(G) \cdot (pR - C \sqrt{pR \log(N)}) \leq \alpha(\widetilde{H}) \leq \alpha(G) \cdot (pR + C \sqrt{pR \log(N)}).
\]
By the choice of $R$ we have $R >\frac{C^2 \lg(N)}{\eps^2 p}$, and hence
$|\alpha(\widetilde{H}) - \alpha(G) pR| \leq  \eps \cdot \alpha(G) pR$.
Therefore, denoting by $m$ the number of vertices in $\widetilde{H}$
if $\alpha(G) > (1/2-\eps) N$, then
$\alpha(\widetilde{H}) \geq (1/2-3\eps) m$,
and hence $\widetilde{H}$ contains a vertex cover of size $(1/2+3\eps) m$.
On the other hand, if $\alpha(G) < \eps N$, then with high probability
$\alpha(\widetilde{H})  < 3\eps m$, and
and hence $\widetilde{H}$ does not have a vertex cover of size $(1-3\eps) m$.
\end{proof}

%%%%%%%%%%%%%%%%%%%%%%%%%%%%%%%%%%%%%%%%%%%%%%%%%%%%%%%
\section{Hamiltonicity and Percolation}
%%%%%%%%%%%%%%%%%%%%%%%%%%%%%%%%%%%%%%%%%%%%%%%%%%%%%%%

Recall that an Hamiltonian cycle in a graph is a cycle that visits every vertex
exactly once. Deciding if a graph (whether directed or undirected)
contains a Hamiltonian cycle is a classical $\NP$-hard problem, which we denote by $\HamCycle$. A hamiltonian path, is a simple path that traverses all vertices in the graph.

%Hamiltonicity may seem like a fragile property that may easily
%break down when edges or vertices are deleted at random.
%However, there are results that show that certain classes
%of Hamiltonian graphs (e.g., Dirac graphs) are fairly robust
%to random deletions and maintain Hamiltonicity even when edges
%are deleted with probability that tends to zero as the size of
%the graph tends to infinity, see, e.g.~\cite{KLS14}.

In this section we prove unless $\NP=\coRP$,
there is no polynomial time algorithm that given a $n$-vertex directed graph $G$
decides with high probability whether $G_{p,e}$ contains a Hamiltonian cycle
for any $p>\frac{1}{n^{1-\epsilon}}$ where $\epsilon \in (0,1)$.

A natural approach in proving that deciding the Hamiltonicity of percolated instances is hard, is to ``blow up'' edges. Namely to replace each edge $(u,v)$ by a clique of size $k$ and connect both endpoints of the edges to all vertices of the clique. The idea is that when $k$ is large enough, there is a Hamiltonian path with high probability between all pairs of distinct vertices of the clique. Hence with high probability, we can connect $u$ and $v$ after percolating the edges, by a path that traverses all the vertices of the percolated clique. The problem with this idea, is that the resulting graph after this blowup operation may not be Hamiltonian as there is a new set of vertices for every ``edge" in the original graph that needs to be traversed by an Hamiltonian cycle. For \emph{directed} graphs, we overcome this problem by adding to each vertex $v$ a large clique $C$, adding a directed edge $(v,c)$ for every $c \in C$ and and adding a directed edge $(c,u)$ for every $c \in C$ and $u \in N(v)$ (where $N(v)$ is the set of all vertices having a directed edge from $v$).

\begin{theorem}\label{thm:ham cycle edge perc}
Let $\eps \in (0,1)$ be a fixed constant.
Then, unless $\NP=\RP$,
there is no polynomial time algorithm that when given a
directed graph $G = (V,E)$ with $n$ vertices
decides with high probability whether $G_{p,e}$
contains a Hamiltonian cycle for any $p > \frac{1}{n^{1-\eps}}$.
\end{theorem}

We will need the following claim.
\begin{claim}\label{claim:ham path gadget}
	Let $H = (V,E)$ be the directed graph
	with a source $s$ a sink $t$, and $R$ vertices $U = \{u_1,\dots,u_k\}$.
	The edges of $H$ are
	\[
		E = \{(s \to u_i) : i \in [R]\}
            \cup \{(u_i \to t) : i \in [R]\}
            \cup \{(u_i \to u_j) : i,j \in [R]\}.
	\]
	Let $H'=(V,E')$ be an edge percolation of $H$,
    where we keep each directed edge with probability $p = \frac{3\log^5(R)}{R}$.
	Then, with probability $1 - \frac{1}{R^3}$ there is a Hamiltonian path in $H$ from $s$ to $t$.
\end{claim}

\begin{proof}
    Let $p_0 \in (0,1)$, and consider the random graph $H_{p_0}$.
	Note that with probability at least
%    $1 - 2(1-p_0)^R-R(p_0(1-p_0)^{R-1})^2$ \inote{why the last term??}
    $1 - 2(1-p_0)^R$
    there are two distinct vertices $v_1,v_R \in U$
	such that $(s \to v_1),(v_R \to t) \in E'$. Conditioning on these specific $v_1,v_R \in U$,
    we show that with high probability there is a Hamiltonian path from $v_1$ to $v_R$
    in the subgraph of $H_{p_0}$ induced by $U$.

    By a result of~\cite[Theorem 1.3]{Hamcyclepacking}
    if $D$ is a $p_0$-edge percolation of the complete directed graph
    with $R$ vertices with $p_0 = \frac{\log^4(R)}{R}$,
    then with high probability
    all edges of $D$ are contained in some Hamiltonian cycle in $D$.
    Note that the probability that $H_{p_0}$ contains
    a Hamiltonian path from $v_1$ to $v_R$ is equal to the probability
    that $H_{p_0}$ contains a Hamiltonian cycle that goes through the edge $(v_1 \to v_R)$,
    conditioned on the event that $(v_1 \to v_R) \in E'$.
    Therefore, since the distribution of the subgraph of $H_{p_0}$ induced by $U$
    is distributed like $D$, it follows that with high probability
    the subgraph $H_{p_0}$ induced by $U$ contains a Hamiltonian path from $v_1$ to $v_R$,
    and hence $H_{p_0}$ contains a Hamiltonian path from $s$ to $t$ with probability at least 1/2.

    Next, let $t = 3 \log(R)$, and let $p = t \cdot p_0$.
    We claim that the graph $H_p$ contains an Hamiltonian path from $s$ to $t$
    with probability at least $\frac{1}{R^3}$.
    Observe that if $H'_1,\dots,H'_t$ are independent copies of $H_{p_0}$,
    then the probability that none of the $H'_i$ contains a
    Hamiltonian path from $s$ to $t$ is at most $(1/2)^t < \frac{1}{R^3}$.
    Therefore, since each edge of $H$ is contained in $\cup_{i=1}^t H_i$
    independently with probability $1 - (1-p_0)^t \leq p$ it follows that
    $H_p$ contains an Hamiltonian path from $s$ to $t$
    with probability at least $\frac{1}{R^3}$, as required.
\end{proof}

\begin{proof}[Proof of Theorem~\ref{thm:ham cycle edge perc}]
	In order to prove the theorem, we show a reduction that given a directed graph
    $G = (V,E)$
    produces a directed graph $G' = (V',E')$ such that
    \begin{itemize}
    \item If $G$ contains a Hamiltonian cycle, then $G'$ contains a Hamiltonian cycle,
    and with high probability $G'_{p,e}$ contains a Hamiltonian cycle.
    \item If $G$ does not contain a Hamiltonian cycle, then neither does $G'$, and hence $G'_{p,e}$ does not contain a Hamiltonian cycle.
    \end{itemize}
    The reduction works as follows.
    Let $V = [N]$ be the vertices of $G$, and let $R$ be a parameter to be chosen later.
    The vertices of $G'$ will be $V' = V \bigcup( \cup_{i=1}^N U_i )$,
    where $U_i = \{u_1^i,\dots,u_R^i\}$.
    For each $i \in [N]$ the graph $G'$ contains all edges in both directions inside $U_i$.
    For each directed edge $(i \to j) \in E$ we add in $G'$ the directed edges
    \[	
    	\{(i \to u_\ell^i) : \ell \in [R]\} \cup \{(u_\ell^i \to j) : \ell \in [R]\}.
    \]
    That is, we turn the graph $G$ into $G'$
    by adding a clique $U_i$ for each vertex $v_i \in V$,
    and letting all edges outgoing from $v_i$ go through this clique.
    This completes the description of the reduction.

    Let us first show that that $G$ contains a Hamiltonian cycle
    if and only if $G'$ contains a Hamiltonian cycle.
    Indeed, suppose that $C = (\sigma_1,\dots,\sigma_N)$ is a Hamiltonian cycle
    in $G$.
    Then
    $C' = (\sigma_1,u_1^{\sigma_1}\dots,u_R^{\sigma_1},
        \dots ,\sigma_N,u_1^{\sigma_N}\dots,u_R^{\sigma_N})$
    is a Hamiltonian cycle in $G'$.
    In the other direction, suppose that $G'$ contains a Hamiltonian cycle $C'$.
	It is easy to see that any $i \in V$ appearing in $C'$
	must be followed immediately by a permutation of all $R$ vertices in $U_i$.
	Therefore, by restricting $C'$ to the vertices in $V$ we get a Hamiltonian cycle in $G$.
	
	Next we show that the reduction above is robust to edge percolation.
	Let $\tilde{G'} = {G'}_{p,e}$ be the edge percolation of $G'$.
	Clearly if $G'$ does not contain a Hamiltonian cycle, then neither does $\tilde{G'}$.
	Therefore, it is only left to show that if $G'$ contain a Hamiltonian cycle $C$,
	then with high probability $\tilde{G'}$ also contains a Hamiltonian cycle.
	As explained above a Hamiltonian cycle in $G'$
	is given by a permutation $\sigma =(\sigma_1,\dots \sigma_N)\in S_N$
	and an arbitrary ordering of the vertices in each $U_i$,
	i.e.,
    $C' = (\sigma_1,u_1^{\sigma_1}\dots,u_R^{\sigma_1}, \dots ,\sigma_N,u_1^{\sigma_N}\dots,u_R^{\sigma_N})$.
    Note that for each $i \in [N]$ the vertices
    $\{\sigma_i,u_1^{\sigma_i}\dots,u_R^{\sigma_i}, \sigma_{i+1}\}$ induce a subgraph isomorphic to the graph $H$
    from Claim~\ref{claim:ham path gadget}.
	Therefore, by Claim~\ref{claim:ham path gadget}
	if $p > \frac{\log^4(R) }{R}$, then for each $i \in [N]$
	with probability $1 - \frac{1}{R^3}$ there is path from $\sigma_i$ to $\sigma_{i+1}$
	that visits all vertices in $U_{\sigma_i}$.
	By taking union bound over all $i \in [N]$ we get that
	with probability $1-\frac{N}{R^3}$ such paths exist for all $i \in [N]$,
	and by concatenating them we conclude that $\tilde{G'}$
	contains a Hamiltonian cycle with high probability.

    Finally, we specify the choice of the parameter $R$.
    The obtained graph $H$ has $n = NR$ vertices,
    and the constraints we have are $p > \frac{\log^4 R}{R}$ and $R^3 \gg N$.
    Therefore, in order to prove the theorem for $p > \frac{1}{n^{1-\eps}}$
    with $\eps \in (0,1)$ it is enough to take $R = N^{1/c}$,
    where $c = \frac{\log(pn)}{\log(n)} > \eps$ such that $p = \frac{1}{n^{1-c}}$.
\end{proof}

%%%%%%%%%%%%%%%%%%%%%%%%%%%%%%%%%%%%%%%%%%%%%%%%%%%%%%%
\section{Constraint Satisfaction Problem and Percolation}
%%%%%%%%%%%%%%%%%%%%%%%%%%%%%%%%%%%%%%%%%%%%%%%%%%%%%%%

In this section we deal with percolation on Constraint Satisfaction Problems ($\CSP$).
An instance $\Phi$ of Boolean $\kCSP[k]$
is a formula consisting of a collection of clauses $C_1,...,C_m$ over $n$ Boolean variables $x_1,...,x_n$,
where each clause is associated with some $k$-ary predicate $f:\{0,1\}^k \rightarrow \{0,1\}$
over $k$ variables $x_{i_1},\dots,x_{i_k}$. An instance $\Phi$ is said to be simple of all clauses in $\Phi$ are distinct.
Given an assignment $\sigma : \{x_1,...,x_n\} \to \{0,1\}$
we say that the constraint $C$ on the variables $x_{i_1},\dots,x_{i_k}$
is a satisfied by $\sigma$ if $f_C(\sigma(x_{i_1},...,\sigma(x_{i_k})) = 1$,
where $f_C$ is the predicate corresponding to $C$.
Given a formula $\Phi$, and an assignment $\sigma$ to its variables
the value of $\Phi$ with respect to the assignment $\sigma$,
denoted by $\val_\sigma(\Phi)$, is fraction of constraints of $\Phi$ satisfied by $\sigma$.
The value of $\Phi$ is defined as $val(\Phi) = \max_{\sigma}\val_\sigma(\Phi)$.
If $\val(\Phi) = 1$ we say that $\Phi$ is satisfiable.

We are typically interested in $\CSP$ where constraints belong
to some fixed family of predicates $\F$.
For example, in the $\kSAT[k]$ problem, the constraints are all of
the form $f(z_1,\dots,z_k) = \bigvee_{i=1}^k (z_i = b_i)$,
for $b_1,\dots,b_k \in \{0,1\}$.
We assume that $k$, the arity of the constraints, is some fixed constant
that does not depend on the number of variables $n$.

These definitions give rise to the following optimization problem.
Given a $\CSP$ instance $\Phi$ find an assignment that maximizes the value of $\Phi$. We refer to this
maximization problem as $\maxCSPF$, where $\F$ denotes the family of predicates constraints are taken from.
For $0 < s < c \leq 1$, let $\gapCSPF(c,s)$ be the promise problem whose
YES-instances are formulas $\Phi$ such that $\val(\Phi) \geq c$,
and NO-instances are formulas $\Phi$ such that $\val(\Phi) \leq s$.
Here we assume the constraints of $\CSP$ instances are restricted
to be in some family $\F$.

We study two models of percolation on instances of $\CSP$.
In \emph{clause percolation} given an instance $\Phi$ of $\CSP$
its clause percolation is a random formula $\Phi^c_p$
over the same set of variables, that is obtained from $\Phi$
by keeping each clause of $\Phi$ independently with probability $p$.

In \emph{variable percolation} given an instance $\Phi$ of $\CSP$
the variable percolation is a random formula $\Phi^v_p$
whose set of variables is a subset $S$ of the variables of $\Phi$,
where each variable of $\Phi$ is in $S$ independently with probability $p \in(0,1)$
and the clauses of $\Phi^c_p$ are all clauses of $\Phi$ induced by $S$.
In other words, a clause $C$ of $\Phi$ survives if and only if all variables from $C$ the percolation process.

\paragraph{Clause percolation}
In this section we show that for Constraint Satisfaction Problem with a $k$-ary constraints,
the problem of approximating the optimal value on percolated instances
is essentially as hard as approximating it on a worst-case instance as long as
$p > \frac{1}{n^{k-1-\delta}}$ for any constant $\delta > 0$.

\begin{theorem}\label{thm:CSP clause perc}
Let $\eps,\delta \in (0,1)$ be fixed constants.
There is a polynomial time reduction
such that given a simple unweighted instance $\Phi$
outputs a simple unweighted instance $\Psi$ on $N$ variables
with the same constraints, such that $\val(\Psi) = \val(\Phi)$,
and furthermore for any $p>\frac{1}{n^{k-1-\delta}}$ the following holds.
\begin{enumerate}
    \item If $\val(\Phi)=1$, then $\val(\Psi^c_p)=1$ with probability 1.
    \item If $\val(\Phi)<1$, then with high probability $|\val(\Psi^c_p)- \val(\Phi)| < \eps$.
\end{enumerate}
\end{theorem}

Theorem~\ref{thm:CSP clause perc} immediately implies the following corollary.

\begin{corollary}\label{cor:CSP perc clause}
Let $\F$ be a collection of Boolean constraints of arity $k$,
and suppose that for some $0 < s < c \leq 1$ the problem $\gapCSPF(c,s)$ is $\NP$-hard.
Then $\gapCSPF(c-\eps,s+\eps)$ is $\NP$-hard under a robust reduction
with respect to clause percolation with any parameter $p>\frac{1}{n^{k-1-\delta}}$,
where $n$ denotes the number of variables in a given formula,
and $\eps, \delta>0$ are arbitrary constants.
\end{corollary}

As a particular application, we have the following result regarding
the hardness of approximating $\kSAT[3]$ on clause-percolated instances.

\begin{theorem}\label{thm:3SAT clause perc}
Let $\eps,\delta \in (0,1)$ be fixed constants.
Then, unless $\NP \seq \coRP$,
there is no polynomial time algorithm that when given a satisfiable instance $\Phi$
of $\kSAT[3]$ finds an assignment $\sigma$ to $\Phi^c_p$ such that
$\val_\sigma(\Phi^c_p) > 7/8 + \eps$ for all $p>\frac{1}{n^{2-\delta}}$.
\end{theorem}

\begin{proof}
By the result of H{\aa}stad~\cite{Hastad01}
for any constant $\eps>0$ given a weighted $\kSAT[3]$ instance $\phi$
it is $\NP$-hard to distinguish between the case that that $\phi$
is satisfiable, and the case that $\val(\phi) < 7/8 + \eps$.
Combining this result with the result of~\cite{weighted}
we get that the same problem is $\NP$-hard also for unweighted simple instances.
The proof follows by applying Theorem~\ref{thm:CSP clause perc}.
\end{proof}

We now return to the proof of Theorem~\ref{thm:CSP clause perc}.

\begin{proof}[Proof of Theorem~\ref{thm:CSP clause perc}]
 We start with the following lemma.
\begin{lemma}\label{lemma:3SAT clause perc}
Let $\Phi$ be a simple unweighted a $\kCSP[k]$ instance with $n$ variables and $m$ clauses,
and let $p > \frac{Cn}{\eps^2 m}$ for some $\eps \in (0,1)$ and
some absolute constant $C > 0$. Then,
\begin{enumerate}
    \item If $\val(\Phi)=1$, then $\val(\Phi^c_p)=1$.
    \item $\val(\Phi)<1$, with high probability $|\val(\Phi^c_p) - \val(\Phi)| < \eps$.
\end{enumerate}
\end{lemma}

\begin{proof}
    The first item is clear, as any assignment that satisfies $\Phi$ will also satisfy $\Phi^c_p$.
    For the second item, denote by $m'$ the number of clauses in $\Phi^c_p$.
	By concentration bounds we have
	\[
		\Pr[|m' - pm| > \eps pm] < e^{-\Omega(\eps^2 pm)} < e^{-\Omega(Cn)},
	\]
	where $\Omega(\cdot)$ hides some absolute constant.
	Fix an assignment $\sigma$ to the variables of $\Phi$, and let $s = \val_\sigma(\Phi)$.
	Then, the number of clauses in $\Phi$ satisfied by $\sigma$ is $s m$.
	Denote by $S_\sigma$ the number clauses in $\Phi^c_p$ satisfied by $\sigma$.
	Since we pick each clause with probability $p$ independently, we have
	\[
		\Pr[ |S_\sigma - s pm |> \eps pm ] < e^{-\Omega(\eps^2 pm)} < e^{-\Omega(Cn)},
	\]
	and hence
	\begin{eqnarray*}
		\Pr[ |\val_\sigma(\Phi^c_p) - s| > \eps ]
		& = & \Pr[ |S_\sigma - s m'| > \eps m' ] \\
		& \leq & \Pr[|m' - pm| > \eps pm/2] + \Pr[| S_\sigma - s pm | >  \eps pm/2 ] \\
		& \leq & 2 e^{-\Omega(Cn)},
	\end{eqnarray*}	
	where $\Omega(\cdot)$ hides some absolute constant.

	Suppose now that that $\val(\Phi) = s$.
	If $\sigma$ is an optimal assignment to $\Phi$, i.e.,
	$\val_\sigma(\Phi) = s$, then we immediately have by the argument above that $\val_\sigma(\Phi^c_p) > s - \eps$ with high probability.
	On the other hand, for any assignment $\sigma'$ it holds that
	$\Pr[ \val_{\sigma'}(\Phi^c_p) > s + \eps ] < e^{-\Omega(Cn)}$
	for some sufficiently large $C>0$, and by
	taking union bound over all assignments $\sigma$ we get
	\[
		\Pr[ \val(\Phi^c_p) > s + \eps ] <
		\Pr[\exists \sigma' \textrm{ such that } \val_{\sigma'}(\Phi^c_p) > s + \eps ] < c^n
	\]
	for some absolute constant $c < 1$.
\end{proof}

We note that we assume in the proof above that $s$ is a constant independent of $n$.
This assumption is justified as $s \geq \frac{1}{2^k}$, and we assume that $k$ is independent of $n$.

Next, we show a polynomial time reduction such that given a $\maxCSPF$ instance $\Phi$ outputs a
$\maxCSPF$ instance $\Psi$ with $N$ variables and $N^{k-\eps}$ clauses such that
$\val(\Psi) = \val(\Phi)$. We use similar ideas to those used in \cite{weighted}
in proving that unweighted instances of CSP problems are as hard to approximate as in the weighted case.

\begin{lemma}\label{lemma:dense CSP}
For any $\delta \in (0,1)$ there is a polynomial time reduction
such that given a simple unweighted $\maxCSPF$ instance $\Phi$ outputs a simple
$\maxCSPF$ instance $\Psi$ with the same constraint
with $n$ variables and at least $n^{k-\delta}$ clauses such that $\val(\Psi) = \val(\Phi)$.
\end{lemma}

\begin{proof}
The reduction works as follows. Let $R$ be a parameter to be chosen later.
Given an instance $\Phi$ of $\kCSP[k]$ with $M$ clauses
over the variables $x_1,\dots,x_N$ the reduction creates the following instance $\Psi$.
For each variable $x_i$ of $\Phi$, the instance $\Psi$
will have a set of $R$ corresponding variables $X_i = \{x_{i,j} : j \in [R]\}$,
where we think of each variable in $X_i$ as a copy of $x_i$.
For each clause $C$ of $\Phi$
we add to $\Psi$ a \emph{cloud} of $R^k$ corresponding clauses,
by taking all possible combinations of the variables from the corresponding $X_i$'s.
We call the set of $R^k$ clauses corresponding to $C$ the \emph{cloud} of $C$.
That is, $\Psi$ has $n = NR$ variables and $m = M \cdot R^k$ clauses.
Therefore, if $R > N^{k/\delta}$, then $m > n^{k-\delta}$.
%DETAILS: $n = N^{k/\delta + 1}$ and
%$m > R^k > N^{k^2/\delta} > N^{(k/\delta + 1 )(k-\delta)} = n^{k-\delta}$.

Next we claim that $\val(\Phi) = \val(\Psi)$.
Clearly, we have $\val(\Phi) \leq \val(\Psi)$,
as any assignment $\sigma : \{x_1,\dots,x_N\} \in \{0,1\}$
to $\Phi$ can be extended to the assignment $\tau$ to $\Psi$
by letting $\tau(x_{i,j}) = \sigma(x_i)$ for all $i \in [N],j \in [R]$.

In the other direction, let $\tau$ be an assignment to the variables of $\Psi$.%
\footnote{Note that if for each $i \in [N]$ the assignment $\tau$ gave
the same value to all variables in $X_i$, this would naturally
induce a corresponding assignment to $\Phi$. However, this need not be the case in general.}
For each $i \in [N]$ let $p_i^1 = \frac{|\{j \in [R] : \tau(x_{i,j} = 1)\}|}{R}$
be the fraction of $x_{i,j}$'s that are assigned the value 1,
and let $p_i^0 = 1 - p_i^1$.
be the fraction of $x_{i,j}$'s that are assigned the value 0.
Construct a assignment $\sigma$ to the variables of $\Phi$ randomly, by setting
$\sigma(x_i) = 1$ with probability $p_i^1$,
and setting $\sigma(x_i) = 0$ with probability $p_i^0$ independently.
Equivalently we choose one of the $R$ copies of $x_i$ in $\Psi$ uniformly at random
and assign to $x_i$ the value assigned by $\tau$ to the variable chosen.
A moment of thought reveals that for each clause $C$ of $\Phi$,
the probability that $\sigma$ satisfies $C$ is equal to the fraction
of the clauses in $\Psi$ in the cloud corresponding to $C$ that are satisfied by $\tau$.
Denote by $SAT_\sigma(C_i)$ the number of clauses that are satisfied by $\sigma$
in the cloud corresponding to $C_i$.
Since each clause of $\Phi$ corresponds to the same number of clauses in $\Psi$,
it follows that the expected value of $\Phi$ under the assignment $\sigma$ is
\begin{eqnarray*}
	\E[\val_\sigma(\Phi)]
	& = & \frac{1}{M} \sum_{i=1}^M \Pr[\textrm{$\sigma$ satisfies $C_i$}] \\
	& = & \frac{1}{M} \sum_{i=1}^M \frac{SAT_\sigma(C_i)}{R^3} \\
	& = & \val_\tau(\Psi).
\end{eqnarray*}
Hence, there exists an assignment $\sigma$ to the variables of $\Phi$
such that $\val_\sigma(\Phi) \geq \val_\tau(\Psi)$,
and thus $\val(\Phi) \geq \val(\Psi)$, as required.
\end{proof}

Theorem~\ref{thm:CSP clause perc} follows immediately from Lemmas~\ref{lemma:3SAT clause perc}
and~\ref{lemma:dense CSP}.

\end{proof}
\medskip
We observe that it is unlikely that Lemma~\ref{lemma:dense CSP} could be generalized to $\maxCSPF$ instances with-arity $k$ and $\Omega(n^k)$ constraints. For example, the value of a $\kSAT$ formula with $\Omega(n^3)$ clauses, admits $1-\delta$ approximation for every $\delta \in (0,1)$ in polynomial time \cite{Arora}.

\paragraph{Variable percolation}
Next we show that $\maxCSPF$ is also hard under variable percolation.
We prove below that for $p$ that is no too small,
with high probability $\maxCSPF$ is hard to approximate on
percolated instances within the same factor as in the worst-case setting.

\begin{theorem}\label{thm:CSP var perc}
Let $\eps,\delta > 0$ be fixed constants.
There is a polynomial time reduction
such that given a simple unweighted instance $\Phi$
outputs a simple unweighted instance $\Psi$ on $n$ variables
with the same constraints, such that $\val(\Psi) = \val(\Phi)$,
and furthermore for any $p>\frac{1}{n^{1-\delta}}$ the following holds.
\begin{enumerate}
    \item If $\val(\Phi)=1$, then $\val(\Psi^v_p)=1$ with probability 1.
    \item If $\val(\Phi)<1$, then with high probability $|\val(\Psi^v_p)- \val(\Phi)| < \eps$.
\end{enumerate}
\end{theorem}

The following corollary is the analogue of Corollary~\ref{cor:CSP perc clause} for variable percolation.

\begin{corollary}\label{cor:CSP perc vertex}
Let $\F$ be a collection of Boolean constraints of arity $k$,
and suppose that for some $0 < s < c \leq 1$ the problem $\gapCSPF(c,s)$ is $\NP$-hard.
Then $\gapCSPF(c-\eps,s+\eps)$ is $\NP$-hard under a robust reduction
with respect to vertex percolation with any parameter $p>\frac{1}{n^{1-\delta}}$,
where $n$ denotes the number of variables in a given formula,
and $\eps, \delta>0$ are arbitrary constants.
\end{corollary}

\begin{proof}[Proof of Theorem~\ref{thm:CSP var perc}]
The reduction is the same reduction as in the proof of Theorem~\ref{thm:CSP clause perc}.
Namely, given a simple unweighted instance $\Phi$ with $N$ variables and $M$ clauses
the reduction replaces each variable $x_i$ of $\Phi$,
with a set of $R$ corresponding variables $X_i = \{x_{i,j} : j \in [R]\}$,
and replaces each clause of $\Phi$ with a \emph{cloud} of $R^k$ corresponding clauses,
by taking all possible combinations of the variables from the corresponding $X_i$'s.
That is, the output of the reduction $\Psi$ has $n = NR$ variables and $m = M \cdot R^k$ clauses.
We choose $R =  N^{1/c}$, where $c = \frac{\log(pn)}{\log(n)} \in (\delta, 1)$
so that $\sqrt{\frac{\log(N)}{pR}} < \frac{1}{N^{c/2}}$.

For each $i \in [N]$ let $X'_i$ be variables from $X_i$ that remain in $\Psi^v_p$
after variable percolation. By concentration, it follows that
for $p > \frac{1}{N^{1-\delta}}$ with high probability
$||X'_i| - pR|< O(\sqrt{pR \log(n)})$ for all $i \in [N]$.
We assume from now on that this is indeed the case.
For a constraint $C_i$ of $\Phi$ let $x_{i_1},\dots,x_{i_k}$ be the
variables that participate in $C_i$.
Then, the number of clauses in the cloud corresponding to $C_i$ in $\Psi^v_p$
is equal to $|X'_{i_1}| \cdots |X'_{i_k}|$, and the total number of clauses in $\Psi^v_p$
is $\sum_{i=1}^M |X'_{i_1}| \cdots |X'_{i_k}|$.

By Lemma~\ref{lemma:dense CSP} we have $\val(\Psi) = \val(\Phi)$.
In particular, if $\Phi$ is satisfiable, then so if $\Psi$, as any assignment that
satisfies $\Psi$ also satisfies any subformula of $\Psi$, which implies that
$\Psi^v_p$ is also satisfiable with probability 1.

Suppose now that $\val(\Phi) < 1$.
We claim that with high probability $|\val(\Psi^v_p) - \val(\Phi)| < \eps$.

To prove that $\val(\Psi^v_p) \geq \val(\Phi) - \eps$,
let $\sigma$ be an optimal assignment to $\Phi$.
Extend $\sigma$ to an assignment $\tau$ to $\Psi^v_p$
by letting $\tau(x_{i,j}) = \sigma(x_i)$ for all $1\le i \le R$.
Note that for each constraint $C_i$ of $\Phi$ if $C_i$ is satisfied by $\sigma$,
then in $\Psi^v_p$ all clauses in the corresponding cloud are satisfied,
and otherwise no clause in the corresponding cloud is satisfied.
Denoting by $SAT_\tau(C_i)$ the number of clauses that are satisfied by $\tau$
in the cloud corresponding to $C_i$ we have
\[
    \val_\tau(\Psi^v_p)
	=\frac{\sum_{i=1}^M {SAT_\tau(C_i)}}{\sum_{i=1}^M |X'_{i_1}| \cdots |X'_{i_k}|}
	\geq\frac{\val(\Phi) M \cdot (p R - \sqrt{pR \log(N)})^k}{M(p R + \sqrt{pR \log(N)})^k}
    \geq \val(\Phi) - O(\sqrt{\frac{\log(N)}{pR}}).
\]
By the choice of $R$ we get for large enough $N$
\[
	\val_\tau(\Psi^v_p) \geq \val(\Phi) - O(\frac{1}{N^{c/2}})
    \geq \val(\Phi) - \eps.
\]

Next, we prove that $\val(\Phi) \geq \val(\Psi^v_p) - \eps$.
Given an assignment $\tau$
to the variables of $\Psi^v_p$ we decode it into an assignment to $\Phi$ using
the same decoding as in the proof of Lemma~\ref{lemma:dense CSP}.
Namely, we choose a random assignment $\sigma$ to the variables of $\Phi$ by setting
$\sigma(x_i)=1$ with probability $p_i^1$ and $\sigma(x_i)=0$ with probability $p_i^0$ independently between $i$'s,
where $p_i^1 =\frac{|\{x_{i,j} \in X'_i : \tau(x_{i,j}=1)\}|}{|X'_i|}$,
and $p_i^0 = 1 - p_i^1$.
Let $C'_i$ be the set of clauses in $C_i$ that belong to $\Psi_{p}^v$. Let $SAT_\tau(C'_i)$ the number of clauses that are satisfied by $\tau$
in $C'_i$,
it follows that the expected value of $\Phi$ under the assignment $\sigma$ is
\begin{equation}\label{eq:val(phi)}
	\E[\val_\sigma(\Phi)]
        = \frac{1}{M} \sum_{i=1}^M \Pr[\textrm{$\sigma$ satisfies $C'_i$}] \\
        = \frac{1}{M} \sum_{i=1}^M \frac{SAT_\tau(C'_i)}{|X'_{i_1}| \cdots |X'_{i_k}|}.
\end{equation}
On the other hand we have
\begin{equation}\label{eq:val(psi perc)}
    \val_\tau(\Psi^v_p)
	=\frac{\sum_{i=1}^M {SAT_\tau(C'_i)}}{\sum_{i=1}^M |X'_{i_1}| \cdots |X'_{i_k}|}.
\end{equation}
Now, using the assumption that for all $i \in[n]$
it holds that $||X'_i| - pR|< \sqrt{pR \log(n)}$,
we get that both \eqref{eq:val(phi)} and \eqref{eq:val(psi perc)} are between
$\frac{\sum_{i=1}^M SAT_\tau(C'_i)}{M(p R + \sqrt{pR \log(N)})^k}$
and
$\frac{\sum_{i=1}^M SAT_\tau(C'_i)}{M(p R - \sqrt{pR \log(N)})^k}$.
A simple computation reveals that the difference between the two quantities is
at most $O(\sqrt{\frac{\log(N}{pR}})$, and hence
\[
	\E[\val_\sigma(\Phi)] \geq \val_\tau(\Psi^v_p) - O(\sqrt{\frac{\log(N)}{pR}})
	\geq \val_\tau(\Psi^v_p) - O(\frac{1}{N^{c/2}})
	\geq \val_\tau(\Psi^v_p) - \eps.
\]
This completes the proof of Theorem~\ref{thm:CSP var perc}.
\end{proof}

%%%%%%%%%%%%%%%%%%%%%%%%%%%%%%%%%%%%%%%%%%%%%%%%%%%%%%%
\section{The Subset Sum Problem and Percolation}
%%%%%%%%%%%%%%%%%%%%%%%%%%%%%%%%%%%%%%%%%%%%%%%%%%%%%%%

In this section we consider the subset-sum problem, and its percolated version.
In the subset-sum problem we are given a set items $\{a_i\}_{i=1}^n$ which are positive integers, and a target integer $S$.
The goal is to decide whether there is a subset of $a_i$'s whose sum is $S$.

Given an instance $I = (\{a_i\}_{i=1}^n;S)$ of the subset sum problem,
we define a percolation on $I$ with probability $p$ to be a random instance $I_p$,
where each item $a_i$ is included in $I_p$ with probability $p$ independently,
with the target of $I_p$ being the same as the target of $I$.

It is known that subset sum is $\NP$-hard.
Below we prove hardness of the percolated version of the subset sum problem.

\begin{theorem}
Rhe $\SS$ problem is $\NP$-hard under robust reduction with respect to
percolation with parameter $p$ for any $p > \frac{1}{n^{1/2-\eps}}$,
where $n$ is the number of items in a given instance, and $\eps > 0$ is any fixed constant.
\end{theorem}

\begin{proof}
	In order to prove the theorem, we show a reduction that given an instance
    $I = (\{a_i\}_{i=1}^N;S)$ of the subset-sum problem with all $a_i > 0$,
    produces an instance $I'$ on $n$ variables such that the following two properties are satisfied.
    \begin{itemize}
    \item If $I \in \SS$, then $I' \in \SS$, and furthermore, with high probability $I'_p \in \SS$.
    \item If $I \notin \SS$, then $I' \notin \SS$, and hence $I'_p \notin \SS$ with probability 1.
    \end{itemize}
    Let us assume that the number of items in $I$ is even.
    (If $N$ is odd, then, add an item to $I$ that is equal to zero).
	Let $R$ be a parameter to be chosen later,
	let $N' = \lceil{\log_2(\sum_i a_i)\rceil}$,
	and for $i=1,\dots,n$ let $M_i = 2^{C'(N' + i)}$ for a large enough constant $C'$.
	For each $i \in [N]$ define the following set
	\[
		J_i = \{ M_i + a_i \cdot N^3 + k : k \in \{-R,\dots,R\}  \}
		\quad
		\textrm{and}
		\quad
		J'_i = \{ M_i + k : k \in \{-R,\dots,R\}  \}
	\]
	Consider now the instance
	\[
		I' = ( \cup_{i \in [N]} (J_i \cup J'_i) ; S'),
	\]
	where $S' = S \cdot N^3 + \sum_{i=1}^N M_i$.
	Clearly this is a polynomial reduction
	that outputs a $\SS$ instance with $n = 2NR$ items.
	
	We show first that $I \in \SS$ if and only if $I'\in \SS$.
	Indeed, suppose that for some subset $T \subseteq [N]$ it holds that $\sum_{i \in T} a_i = S$.
	Consider the following subset of items of $I'$.
	For each $i \in T$ take the item from $J_i$ that corresponds to $k=0$,
	and for $i \in [N] \setminus T$ take the item from $J'_i$ that corresponds to $k=0$.
	Then, by taking these items we are getting
	\[
		\sum_{i \in T} (a_i \cdot N^3 + M_i) + \sum_{i \in [N] \setminus T} M_i = S'.
	\]

	In the other direction, suppose that $I' \in \N$.
	Then, there is some subset $T' \subseteq [N] \times \{0,1\} \times \{-R,\dots,R\}$ such that
	\[
	\sum_{(i,t,k) \in T'} (M_i + a_i \cdot N^3 \cdot t + k)  =  S' = S \cdot N^3 + \sum_{i \in [N]} M_i.
	\]
	Note that by the choices of $M_i$ (namely because $M_i$'s are much larger than $a_i$'s and $R$)
	for each $i \in [N]$ there is a unique $t_i \in \{0,1\}$
	and a unique $k_i \in \{-R,\dots,R\}$ such that $(i,t_i,k_i) \in T'$.
	Therefore, since $\sum_{i=1}^N |k_i| \leq N D < N^3$,
	it follows that $\sum_{i \in [N]} k_i = 0$, and hence
	by defining $T = \{i \in [N] : t_i = 1\}$ we get that
	$\sum_{i \in T} a_i = S$, and so $I \in \SS$.
	
	\medskip
	
	Next, we claim that the reduction above is in fact robust.
	Indeed, consider the percolated instance $I'_p$ for some $p \in (0,1]$.
	Note that if $I \notin \SS$, then neither is $I'$,
    and hence $I'_p \notin \SS$ with probability 1.
    It remains to show that if $I \in \SS$, then
	with high probability $I'_p \in \SS$.
	The proof relies on the following claim.
	\begin{claim}\label{claim:zerosum}
        Let $N \in \mathbb{N}$ be even, and let $R \in \mathbb{N}$.
		Let $A_1,\dots,A_n \subseteq \{-R,\dots,R\}$ be random sets chosen by letting
		each $k \in \{-R,\dots,R\}$ to be in $A_i$ with probability $p$ independently of each other.
		Then, with probability $\geq 1-N/2 \cdot (1-p^2)^{2R}$ for each $i \in [n]$
        there is $k_i \in A_i$ such that $\sum_{i=1}^N k_i = 0$.
	\end{claim}
\begin{proof}
	Note that for each odd $i \in [N]$,
	the probability for a fixed element $x \in \{-R,\dots,R\}$
	that both $x \in A_i$ and $-x \in A_{i+1}$ hold is $p^2$. Therefore,
	\[
	\Pr[\exists k \in \{-R,\dots,R\} : k \in A_i \mbox{ and } -k \in A_{i+1} ]
%	= 1 - \Pr[\forall k \in \{-R,\dots,R\} : \mbox{either } k \notin A_i \mbox{ or } -k \notin A_{i+1}]
	 = 1 -(1-p^2)^{2R+1}.
	\]
	Hence, by taking the union bound over all pairs $(i,i+1)$
	with odd values of $i$ we get that with probability at least  $1-N/2 \cdot (1-p^2)^{2R}$,
	for all odd $i$'s there is $k_i \in A_i$ such that $-k_i \in A_{i+1}$.
\end{proof}

	Suppose now that $I \in \SS$, i.e.,
	for some subset $T \subseteq [N]$ it holds that $\sum_{i \in T} a_i = S$.
	Note that the percolated instance $I'_p$ is obtained from $I'$
	by taking random subsets of $J_i$ and $J'_i$ independently of each other.
	For $i \in [N]$ define $A_i$ to be the $p$-percolated subsets of $J_i$ if $i \in T$,
	and define $A_i$ to be the $p$-percolated subsets of $J'_i$ if $i \notin T$.
	Note that if $R > \frac{C \log(N)}{p^2}$, then the conclusion of
	Claim~\ref{claim:zerosum} holds with probability at least $1 - 1/N$.
	Therefore, in the percolated instance $I'_p$ by taking the items of $I'_p$
	from $A_i$'s that correspond to $k \in A_i$'s from Claim~\ref{claim:zerosum}
	we get
	\begin{eqnarray*}
	\sum_{i \in T} (M_i + a_i \cdot N^3 + k_i) + \sum_{i \in [N] \setminus T} (M_i + k_i)
	& = & (\sum_{i \in [N]} M_i) + (\sum_{i \in T} a_i \cdot N^3) + (\sum_{i \in [N]} k_i) \\
	& = & (\sum_{i \in [N]} M_i) + S + 0 \\
	& = & S'.
	\end{eqnarray*}
	Therefore, with high probability $I'_p \in$ subset-sum as required.

	Finally, note that the reduction works as long as $R > C \frac{\log(N)}{p^2}$,
	or equivalently $p > C\sqrt{\frac{\log N}{R}}$.
	It is easy to verify that for $R = N^{1/c}$, with $c = \frac{\log(p \sqrt{n})}{\log(n)}$,
	the foregoing reduction is indeed a robust reduction
	with respect to percolation with parameter $p > \frac{1}{n^{1/2-\eps}}$
	for any constant $\eps > 0$, where $n$ is the number of items in the $\SS$ instance.
\end{proof}

%%%%%%%%%%%%%%%%%%%%%%%%%%%%%%%%%%%%%%%%%%%%%%%%%%%%%%%
\section{Conclusion}
%%%%%%%%%%%%%%%%%%%%%%%%%%%%%%%%%%%%%%%%%%%%%%%%%%%%%%%

We have examined the complexity of percolated instances of several well known
$\NP$-hard problems and established the hardness of solving exactly and approximately
these problems on such instances. It might be of interest to study percolated instances
of other $\NP$-hard problems that were not considered here.

There are several question arising from this work.
For the $\HamCycle$ problem it would be interesting to
determine whether vertex-percolated instances are hard (in directed or undirected graphs).
Currently, we are unable to establish that this problem is hard even if every vertex remains with probability $p < 1$, where $p$ is a constant that does not depend on the size of the graph.
It could be the case that there is no reduction from $\NP$ to $\HamCycle$ that is
robust to vertex percolation. Proving the inexistence of such reductions (if true)
could be of interest.

\medskip

It might also prove worthwhile to determine whether percolated instances
of $\kSAT[3]$ remain hard to solve even if $p=O(1/n^2)$ over $n$-variable formulas.
Several works suggest that finding a satisfying assignment for a random
$\kSAT[3]$ instance with $n$ variables and $C_1 n$ clauses is hard when
$C_1$ is close to the satisfiability threshold \cite{bart}. Other works
suggest that certifying the unsatisfiability of a random $\kSAT[3]$ instance
with $C_2n$ clauses, with $C_2$ being large enough, is difficult as well~\cite{Feige02}.
These works may serve as evidence that $\kSAT[3]$ should be hard to solve for $p=O(1/n^2)$.

\medskip

One of the first algorithms for solving independent sets in the
random graph $G(n,p)$ is the Karp-Sipser algorithm~\cite{KS}.
This algorithm works by choosing iteratively a degree one vertex randomly,
adding it to the independent set, and removing both the selected vertex
and its sole neighbor from the graph. When there are no vertices of degree one,
the algorithm terminates. It was proven in \cite{Pittel} that when $p<\frac{e}{n}$
(where $e$ is the base of the natural logarithm, $e=2.71828 \ldots$) the algorithm
finds with high probability an optimal independent set. When $p>\frac{e}{n}$,
the Karp-Sipser algorithm fails (with high probability) to find a maximum
independent set of the graph \cite{KS}. Despite extensive research, no algorithm is known
to find an optimal independent set in $G(n,p)$ when $p>\frac{e}{n}$.

It is not clear whether the Karp-Sipser algorithm works
on random subgraphs of worst-case graphs, as opposed to
a random subgraph of the complete graph. This leads to the following problem.

\begin{problem}
	Let $p<\frac{e}{n}$. Is there a polynomial-time algorithm that given an $n$-vertex graph $G$
	finds a maximal independent set in $G_{p,e}$ with high probability?

	Is it true that for $p>\frac{e}{n}$
	no algorithm can find a maximal independent set in $G_{p,e}$ for worst case instance $G$, unless $\NP \subseteq \BPP$?
\end{problem}

\section*{Acknowledgements}
We thank Itai Benjamini, Huck Bennett, Uri Feige and Sam Hopkins for useful discussions.

\bibliography{po}
\bibliographystyle{plain}

\end{document}